\documentclass[twocolumn]{autart}    

\usepackage[T1]{fontenc}
\usepackage[acronym]{glossaries}
\usepackage{amssymb}
\usepackage{subcaption}
\usepackage{tikz}
\usepackage{pgfplots}
\pgfplotsset{compat=1.18}
\usepgfplotslibrary{groupplots}
\usetikzlibrary{patterns}

\usepackage{graphicx}          

\newtheorem{problem}{Problem}
\newtheorem{definition}{Definition}
\newtheorem{theorem}{Theorem}
\newtheorem{lemma}{Lemma}
\newtheorem{remark}{Remark}
\newtheorem{assumption}{Assumption}
\newtheorem{corollary}{Corollary}

\newacronym{lqg}{LQG}{linear--quadratic Gaussian}
\newacronym{lq}{LQ}{linear--quadratic}
\newacronym{isdg}{ISDG}{Inverse Stochastic Differential Game}
\newacronym{irl}{IRL}{Inverse Reinforcement Learning}
\newacronym{hjb}{HJB}{Hamilton-Jacobi-Bellman}
\newacronym{pmp}{PMP}{Pontryagin Minimum Principle}

\begin{document}

\begin{frontmatter}

\title{Infinite-Horizon Inverse Linear--Quadratic Differential Games with State- and Control-Dependent Noise\thanksref{footnoteinfo}} 

\thanks[footnoteinfo]{This paper was not presented at any conference. Corresponding author Lucas G\"unther.} 

\author[KIT]{Lucas G\"unther}\ead{lucas.guenther@kit.edu},
\author[KIT]{Karl Handwerker}\ead{karl.handwerker@kit.edu},            
\author[KIT]{Felix Th\"ommes}\ead{felix.thoemmes@kit.edu},
\author[KIT]{Balint Varga}\ead{balint.varga2@kit.edu},
\author[KIT]{S\"oren Hohmann}\ead{soeren.hohmann@kit.edu}

\address[KIT]{Institute of Control Systems, Karlsruhe Institute of Technology (KIT), Kaiserstr. 12, 76131 Karlsruhe, Germany}

\begin{keyword}                           
Differential Games; Inverse Stochastic Differential Games; Inverse Reinforcement Learning; Identification.               
\end{keyword}                             

\begin{abstract}                          
This paper presents a method to solve the inverse problem for $N$-player infinite-horizon linear--quadratic (LQ) differential games with state- and control-dependent noise. For this stochastic setting, we derive necessary and sufficient conditions for linear feedback Nash equilibria, which take the form of coupled stochastic algebraic Riccati equations. We then derive a kernel representation of these equations to explicitly characterize the set of all cost function parameter combinations across players that are consistent with observed equilibrium trajectories, thereby solving the associated inverse problem. Numerical results illustrate the approach and confirm the theoretical findings, highlighting the inherent ambiguity of the inverse problem.
\end{abstract}

\end{frontmatter}

\section{Introduction}\label{sec:introduction}
Non-cooperative stochastic differential games provide a mathematical framework for modeling dynamic interactions among multiple decision-makers operating under uncertainty \cite{Moo:26}. As a stochastic extension of deterministic differential games \cite{Isa:65}, they arise in a wide range of application domains, including financial markets \cite{Car:16}, economics \cite{Bas:86}, engineering \cite{For:15}, and human–machine interaction \cite{Kil:24}. In many of these applications, uncertainty enters the system through state- and control-dependent noise, for example via stochastic volatility in financial models \cite{Car:16} or human variability in human–machine interaction \cite{TodJor:02}.
The theoretical analysis of stochastic differential games has been studied in the literature \cite{BucLi:08}, \cite{BasOls:99}, \cite{Sun:19}, \cite{Zhu:13}. In this context, the computation of Nash equilibria, also referred to as the forward problem, remains an active area of research and typically assumes that the players’ cost functions are known, from which the resulting equilibrium strategies and trajectories are derived \cite{Kus:02}, \cite{Han:26}.
However, in practical applications, these cost functions are often not directly accessible, leading to the inverse problem of inferring the players' underlying cost functions from observed equilibrium trajectories, commonly referred to as ground truth trajectories.

To date, inverse methods have been studied predominantly in the single-player setting, in particular within the frameworks of Inverse Optimal Control \cite{Kar:24} or \gls{irl} \cite{Zie:08}. These approaches, however, do not capture strategic interactions between multiple players. Extensions to multi-player settings have been considered in the context of multi-agent \gls{irl}. Early contributions include the works of \cite{Nat:10} on cooperative settings and \cite{Red:12} on non-cooperative settings. Subsequent developments have focused on extending maximum entropy \gls{irl} \cite{Zie:08} to multi-player systems, resulting in iterative and approximate solution methods \cite{Meh:23}, \cite{Neu:21}. In parallel, learning-based approaches have been proposed that leverage deep learning techniques, including adversarial formulations \cite{Yu:19}, Q-learning-based methods \cite{Hay:25}, and autoencoder architectures \cite{Sun:25}. While these multi-agent \gls{irl} approaches consider inverse problems with multiple interacting players, they are formulated in discrete-time settings, typically as Markov or dynamic games, and therefore do not apply to continuous-time stochastic differential games.

In continuous-time multi-player settings, inverse problems have been studied in the context of deterministic differential games. Within this line of work, two main methodological approaches can be distinguished: those based on the \gls{pmp} \cite{Mol:22}, and those based on coupled \gls{hjb} equations \cite{Lia:22}. In the deterministic \gls{lq} setting, linear system dynamics and quadratic cost functions imply that both \gls{pmp} and coupled \gls{hjb} equations lead to coupled Riccati differential equations as equilibrium conditions \cite[Corollary 6.5]{BasOls:99}. If an infinite-horizon deterministic \gls{lq} differential game is considered, coupled algebraic Riccati equations arise, which constitute necessary and sufficient conditions for linear feedback Nash equilibria \cite{Eng:00}. Exploiting this property, the inverse methods proposed in \cite{Ing:19}, \cite{MarCao:24} utilize the equilibrium conditions to infer the underlying cost function parameters from observed trajectories or feedback strategies.

In $N$-player non-cooperative stochastic differential games, equilibrium conditions are characterized by $N$ coupled stochastic \gls{hjb} equations. To date, \gls{isdg} methods, which address the corresponding inverse problem, have only been developed for the \gls{lqg} setting. In this case, due to additive Gaussian noise, the equilibrium conditions coincide with those of the deterministic \gls{lq} case \cite[Corollary 6.12]{BasOls:99}. This equivalence allows the corresponding coupled deterministic Riccati differential or algebraic Riccati equations to be utilized within \gls{isdg} methods for finite-horizon \cite{Gue:26}, \cite{Ren:26}, \cite{Li:26} and infinite-horizon \cite{Che:24a}, \cite{Che:25} \gls{lqg} differential games. However, when state- and in particular control-dependent noise is introduced, the equilibrium conditions no longer coincide with their deterministic counterparts, and consequently none of the existing approaches is applicable. For $N$-player infinite-horizon stochastic differential games with state- and control-dependent noise, the coupled stochastic algebraic Riccati equations characterizing linear feedback Nash equilibria were formulated in \cite{Han:26}. However, establishing that these equations are necessary and sufficient conditions for linear feedback Nash equilibria, a prerequisite for employing them within inverse methods, remains an open problem.

Consequently, inverse methods for \gls{lq} differential games with state- and control-dependent noise have not yet been developed. The infinite-horizon setting is particularly promising, as it leads to algebraic equilibrium conditions that can be exploited in inverse problem formulations. To address this gap and solve the inverse problem for infinite-horizon \gls{lq} differential games with state- and control-dependent noise for the first time, this paper makes three main contributions:
\begin{itemize}
    \item Establish the coupled stochastic algebraic Riccati equations as necessary and sufficient conditions for linear feedback Nash equilibria.
    \item Derive a kernel representation of these equilibrium conditions as per-player linear systems of equations in that player's cost function parameters.
    \item Characterize the complete set of cost function parameter combinations across all players that are consistent with observed equilibrium trajectories.
\end{itemize}
\section{Problem Formulation}
\label{sec:problem_formulation}
Let $(\Omega,\mathcal{F},\mathbb{P})$ be a probability space and let $\{\mathcal{F}_t\}_{t\geq 0}$ be a right-continuous filtration such that $\mathcal{F}_0$ contains all $\mathbb{P}$-null sets of $\mathcal{F}$ \cite[Definition~1.5]{Nis:15}. Let $\omega(\cdot)$ be a $\{\mathcal{F}_t\}_{t\geq 0}$-adapted one-dimensional standard Wiener process. Let the system state $\boldsymbol{x}(t) \in \mathbb{R}^n$ be assumed to be $\{\mathcal{F}_t\}_{t\geq 0}$-adapted with deterministic initial condition $\boldsymbol{x}(0)=\boldsymbol{x}_0 \in \mathbb{R}^n$. Let $\boldsymbol{u}_i(t) \in \mathbb{R}^{m_i}$ denote the control input of player $i \in \mathcal{I} := \{1,\dots,N\} \subset \mathbb{N}_+$. Let the linear stochastic system dynamics with state- and control-dependent noise be given by
\begin{equation}\label{eq:system}
\begin{aligned}
    d\boldsymbol{x}(t) &= \bigg(\boldsymbol{A} \boldsymbol{x}(t) + \sum_{j\in\mathcal{I}} \boldsymbol{B}_j \boldsymbol{u}_j(t)\bigg) dt \\
    &+ \bigg(\boldsymbol{C} \boldsymbol{x}(t) + \sum_{j\in\mathcal{I}} \boldsymbol{D}_j \boldsymbol{u}_j(t)\bigg) d\omega(t),
\end{aligned}
\end{equation}
where $\boldsymbol{A} \in \mathbb{R}^{n\times n}$, $\boldsymbol{C} \in \mathbb{R}^{n\times n}$, and $\boldsymbol{B}_i \in \mathbb{R}^{n\times m_i}$, $\boldsymbol{D}_i \in \mathbb{R}^{n\times m_i}$ for all $i \in \mathcal{I}$.
We consider that the players select their controls using constant linear feedback strategies
\begin{equation}\label{eq:feedback_strategy}
    \boldsymbol{u}_i(t) =  - \boldsymbol{K}_i\boldsymbol{x}(t),\quad \forall i \in \mathcal{I},
\end{equation}
with $\boldsymbol{K}_i\in \mathbb{R}^{m_i \times n}$ such that $\boldsymbol{K} := \left(\boldsymbol{K}_1, \dots,\boldsymbol{K}_N\right)$ denotes the $N$-tuple of feedback strategies. Using \eqref{eq:system} and \eqref{eq:feedback_strategy}, we define the closed-loop drift and diffusion matrices
\begin{equation}\label{eq:closed_loop_drift_diffusion}
     \boldsymbol{A}_{cl}(\boldsymbol{K}) := \boldsymbol{A}-\sum_{j\in\mathcal{I}}\boldsymbol{B}_j\boldsymbol{K}_j, \quad \boldsymbol{C}_{cl}(\boldsymbol{K}) := \boldsymbol{C}-\sum_{j\in\mathcal{I}}\boldsymbol{D}_j\boldsymbol{K}_j,
\end{equation}
so that \eqref{eq:system} can be written in closed-loop form as
\begin{equation}\label{eq:closed_loop_system}
    d\boldsymbol{x}(t) = \boldsymbol{A}_{cl}(\boldsymbol{K})\boldsymbol{x}(t) dt
    + \boldsymbol{C}_{cl}(\boldsymbol{K})\boldsymbol{x}(t) d\omega(t).
\end{equation}
\begin{definition}[\cite{Dam:04}, Definition~1.5.1]
    The closed-loop system \eqref{eq:closed_loop_system} is mean-square stable if
    \begin{equation}\label{eq:mean_square_stability}
        \lim_{t \to \infty} \mathbb{E}\left[\|\boldsymbol{x}(t)\|^2\right] = 0, \quad \forall \boldsymbol{x}_0 \in \mathbb{R}^n.
    \end{equation}
    The open-loop system \eqref{eq:system} is mean-square stabilizable if there exists an $N$-tuple $\boldsymbol{K}$ such that \eqref{eq:closed_loop_system} is mean-square stable.
\end{definition}
We denote the set of stabilizing constant linear feedback strategies by
\begin{equation}
    \mathcal{K} := \bigg\{\boldsymbol{K} \in\prod_{i\in\mathcal{I}}\mathbb{R}^{m_i\times n} \bigg| \eqref{eq:closed_loop_system}\text{ is mean-square stable}\bigg\}.
\end{equation}
\begin{assumption}\label{ass:mss}
The system \eqref{eq:system} is mean-square stabilizable, and we restrict our attention to $\boldsymbol{K}\in \mathcal{K}$.
\end{assumption}
We consider non-cooperative stochastic differential games, where each player $i \in \mathcal{I}$ aims to minimize an individual infinite-horizon cost function
\begin{equation}\label{eq:cost_function}
    J_i = \mathbb{E}\bigg[\int_{0}^{\infty}\boldsymbol{x}(t)^\top\boldsymbol{Q}_i\boldsymbol{x}(t) + \sum_{j\in\mathcal{I}} \boldsymbol{u}_j(t)^\top\boldsymbol{R}_{ij}\boldsymbol{u}_j(t)\,dt\bigg],
\end{equation}
subject to \eqref{eq:system}.  Let $\boldsymbol{\theta}_i \in \mathbb{R}^{p_i}$, with $p_i = n^2 + \sum_{j\in\mathcal{I}} m_j^2$, denote the vectorization of the matrices in \eqref{eq:cost_function}, defined as
    \begin{equation}\label{eq:theta_i}
        \boldsymbol{\theta}_i :=
        \big[
            \operatorname{vec}(\boldsymbol{Q}_i)^\top\quad
            \operatorname{vec}(\boldsymbol{R}_{i1})^\top\;\;\cdots\;\;
            \operatorname{vec}(\boldsymbol{R}_{iN})^\top
        \big]^\top,       
    \end{equation}
    where $\operatorname{vec}(\boldsymbol{X})$ denotes the column-wise vectorization of a matrix $\boldsymbol{X}$.
\begin{assumption}\label{ass:admissible_set}
    Let the cost function parameters $\boldsymbol{\theta}_i,\,\forall i \in \mathcal{I}$, $\forall j \in \mathcal{I} \setminus \{i\}$, lie within the admissible set
    \begin{equation}
        \mathcal{A}_i:=\{\boldsymbol{\theta}_i\,\mid\,\boldsymbol{Q}_i \in \mathbb{S}^n_{++},\,\boldsymbol{R}_{ii} \in \mathbb{S}^{m_i}_{++},\,\boldsymbol{R}_{ij} \in \mathbb{S}^{m_j}_{+}\},
    \end{equation}
    where $\boldsymbol{X}\in \mathbb{S}^n_{++}$ $(\boldsymbol{X}\in \mathbb{S}^n_{+})$ denotes that $\boldsymbol{X}$ is a positive definite (positive semidefinite) symmetric $n \times n$ matrix.
\end{assumption}
We write $J_i$ as $J_i(\boldsymbol{x}_0,\boldsymbol{K}_i, \boldsymbol{K}_{\neg i})$, where $\boldsymbol{K}_{\neg i} := ( \boldsymbol{K}_j \mid  j  \in  \mathcal{I}\setminus\{i\})$, as a function of feedback strategies and the initial state, since they generate the state and control trajectories $\boldsymbol{x}(t)$ and $\boldsymbol{u}_i(t),\,\forall i \in \mathcal{I}$, through \eqref{eq:system} and \eqref{eq:feedback_strategy}.
\begin{definition}[\cite{Eng:00}, Definition~1]\label{def:nash_equilibrium}
    An $N$-tuple of feedback strategies $\boldsymbol{K}^\star \in \mathcal{K}$ is called a stabilizing linear feedback Nash equilibrium if 
    \begin{equation}\label{eq:feedback_nash_equilibrium}
    J_i(\boldsymbol{x}_0,\boldsymbol{K}_i^\star, \boldsymbol{K}_{\neg i}^\star) \leq J_i(\boldsymbol{x}_0,\boldsymbol{K}_i, \boldsymbol{K}_{\neg i}^\star), \quad \forall \boldsymbol{x}_0 \in \mathbb{R}^n,
    \end{equation}
    holds for all $i \in \mathcal{I}$ such that $(\boldsymbol{K}_i,\boldsymbol{K}_{\neg i}^\star) \in \mathcal{K}$.
\end{definition}
\begin{definition}\label{def:solution_set}
Given a linear feedback Nash equilibrium $\boldsymbol{K}^\star$, the solution set of the inverse \gls{lq} differential game with state- and control-dependent noise is the set $\boldsymbol{\Theta}$ of all joint parameter vectors
$\hat{\boldsymbol{\theta}}=(\hat{\boldsymbol{\theta}}_1,\dots,\hat{\boldsymbol{\theta}}_N)$
such that, for each $\hat{\boldsymbol{\theta}} \in \boldsymbol{\Theta}$, every
$\hat{\boldsymbol{\theta}}_i$, $i \in \mathcal{I}$, parameterizes the cost function
\eqref{eq:cost_function} of player $i$ and satisfies
Assumption~\ref{ass:admissible_set}, and the induced \gls{lq} differential game parameterized by $\hat{\boldsymbol{\theta}}$ admits
$\boldsymbol{K}^\star$ as a linear feedback Nash equilibrium.
\end{definition}
Let $T_N > 0$ be a finite observation horizon and let $D \in \mathbb{N}_+$ independent demonstrations of the trajectories $\boldsymbol{x}^{(d)}(t)$ and $\boldsymbol{u}_i^{(d)}(t)$, $i \in \mathcal{I}$ and $d \in \{1,\dots,D\}$, be given.
\begin{assumption}\label{ass:nash_trajectories}
Let the demonstrations be generated through \eqref{eq:closed_loop_system} under a fixed linear feedback Nash
equilibrium $\boldsymbol{K}^\star$, as in Definition~\ref{def:nash_equilibrium}, of the infinite-horizon game \eqref{eq:cost_function} with $\boldsymbol{\theta}^\star=(\boldsymbol{\theta}_1^\star, \dots, \boldsymbol{\theta}_N^\star)$, such that the ground truth trajectories satisfy
\begin{equation}\label{eq:nash_trajectories}
  \boldsymbol{u}_i^{\star(d)}(t) = -\boldsymbol{K}_i^\star \boldsymbol{x}^{\star(d)}(t), 
  \quad \forall t \in [0,T_N],\ \forall i \in \mathcal{I},
\end{equation}
almost surely.
\end{assumption}
For the given demonstrations as in Assumption~\ref{ass:nash_trajectories}, define
\begin{equation}\label{eq:data_matrices}
\begin{aligned}
\overline{\boldsymbol{X}}
&:=
\int_0^{T_N} \frac{1}{D}\sum_{d=1}^D 
\boldsymbol{x}^{\star(d)}(t)(\boldsymbol{x}^{\star(d)}(t))^\top dt,\\
\overline{\boldsymbol{U}}_i
&:=
\int_0^{T_N} \frac{1}{D}\sum_{d=1}^D 
\boldsymbol{u}_i^{\star(d)}(t)(\boldsymbol{x}^{\star(d)}(t))^\top dt,
\quad\forall i \in \mathcal{I}.
\end{aligned}
\end{equation}
Under Assumption~\ref{ass:mss}, \eqref{eq:closed_loop_system} is asymptotically mean-square stable, which in this setting is equivalent to exponential mean-square stability \cite[Theorem~1.5.3]{Dam:04}. Therefore, the integrals in \eqref{eq:data_matrices} are well-defined and finite almost surely.
\begin{assumption}\label{ass:persistence_excitation}
Let there exist a constant $\alpha>0$ such that $\overline{\boldsymbol{X}} \succeq \alpha \boldsymbol{I}_n$ almost surely, ensuring persistence of excitation.
\end{assumption}
\begin{definition}\label{def:trajectory_matching}
Let $\hat{\boldsymbol{\theta}}$ be a joint parameter vector as in Definition~\ref{def:solution_set}.
Let $\hat{\boldsymbol{x}}^{(d)}(t)$ and $\hat{\boldsymbol{u}}_i^{(d)}(t)$, $\forall i \in \mathcal{I}$,
denote trajectories corresponding to one of the linear feedback Nash equilibria induced by $\hat{\boldsymbol{\theta}}$
for the $d$-th realization of the underlying Wiener process $\omega(\cdot)$
and the initial condition $\boldsymbol{x}_0$ through \eqref{eq:closed_loop_system}. The trajectories are said to match the ground truth trajectories if, for each $d \in \{1,\dots,D\}$ and for all $t \in [0,T_N]$, it holds almost surely that
\begin{equation}
\boldsymbol{x}^{\star(d)}(t) = \hat{\boldsymbol{x}}^{(d)}(t), \quad
\boldsymbol{u}_i^{\star(d)}(t) = \hat{\boldsymbol{u}}_i^{(d)}(t),\quad \forall i \in \mathcal{I}.
\end{equation}
\end{definition}
\begin{problem}\label{prob:isdg}
    Let Assumptions~\ref{ass:mss},~\ref{ass:admissible_set},~\ref{ass:nash_trajectories}, and~\ref{ass:persistence_excitation} hold. Furthermore, let the system matrices $\boldsymbol{A},\boldsymbol{C}$, and $\boldsymbol{B}_i,\boldsymbol{D}_i,\, \forall i \in \mathcal{I}$, be given. Determine the solution set $\boldsymbol{\Theta}$ according to Definition~\ref{def:solution_set} such that, for every $\hat{\boldsymbol{\theta}} \in \boldsymbol{\Theta}$, there exist trajectories satisfying Definition~\ref{def:trajectory_matching}.
\end{problem}
\section{Inverse Stochastic Differential Games}\label{sec:isdg}
To address Problem~\ref{prob:isdg}, we introduce a novel inverse method for \gls{lq} differential games with state- and control-dependent noise. The approach proceeds in four steps. First, we construct a least-squares estimator to recover the underlying linear feedback strategies $\boldsymbol{K}^\star$ from observed ground truth trajectories. Second, we establish that the coupled stochastic algebraic Riccati equations constitute necessary and sufficient conditions for linear feedback Nash equilibria. Third, we derive a kernel representation of these equilibrium conditions, for each player $i \in \mathcal{I}$, as linear systems of equations in that player's cost function parameters $\boldsymbol{\theta}_i$. Finally, we establish the main result: an explicit characterization of the solution set $\boldsymbol{\Theta}$ as the Cartesian product of the player-wise solution sets of these per-player linear systems.
\subsection{Recovery of Equilibrium Strategies}
\begin{lemma}\label{lem:identification_strategies}
Let Assumptions~\ref{ass:mss}, \ref{ass:nash_trajectories}, and~\ref{ass:persistence_excitation} hold. Then, for each $i \in \mathcal{I}$, the least-squares estimator
\begin{equation}\label{eq:ls_estimator}
    \hat{\boldsymbol{K}}_i = - \overline{\boldsymbol{U}}_i \, \overline{\boldsymbol{X}}^{-1}
\end{equation}
recovers the underlying linear feedback strategy, i.e., $\hat{\boldsymbol{K}}_i = \boldsymbol{K}_i^\star$ almost surely. Moreover, $\hat{\boldsymbol{K}}_i$ generates trajectories that match the ground truth trajectories as in Definition~\ref{def:trajectory_matching}.
\end{lemma}
\begin{pf}
Under Assumptions~\ref{ass:mss} and~\ref{ass:nash_trajectories}, for each $i \in \mathcal{I}$, substituting \eqref{eq:nash_trajectories} into \eqref{eq:data_matrices} yields $\overline{\boldsymbol{U}}_i = - \boldsymbol{K}_i^\star\overline{\boldsymbol{X}}$. By Assumption~\ref{ass:persistence_excitation}, $\overline{\boldsymbol{X}}$ is positive definite and hence invertible almost surely. Therefore, $\hat{\boldsymbol{K}}_i = - \overline{\boldsymbol{U}}_i \, \overline{\boldsymbol{X}}^{-1} = \boldsymbol{K}_i^\star$ almost surely. For any realization of the underlying Wiener process $\omega(\cdot)$ and the initial condition $\boldsymbol{x}_0$, $\hat{\boldsymbol{K}}_i$ and $\boldsymbol{K}_i^\star$ produce trajectories that match as in Definition~\ref{def:trajectory_matching} due to unique solvability of the stochastic differential equation \eqref{eq:closed_loop_system} and Assumption~\ref{ass:mss}.
\qed
\end{pf}

\subsection{Derivation of Equilibrium Conditions}
\begin{lemma}\label{lem:coupled_Riccati_necessary_sufficient}
    Let Assumptions~\ref{ass:mss} and~\ref{ass:admissible_set} hold. Let there exist an $N$-tuple $(\boldsymbol{P}_1,\dots,\boldsymbol{P}_N)$ with $\boldsymbol{P}_i\in \mathbb{S}^n_{++},\,\forall i \in \mathcal{I}$, satisfying the $N$ coupled stochastic algebraic Riccati equations
    \begin{equation}\label{eq:coupled_riccati}
    \begin{aligned}
        \boldsymbol{0}&=\boldsymbol{P}_i \boldsymbol{A}_{cl}(\boldsymbol{K}    ^\star) + \boldsymbol{A}_{cl}(\boldsymbol{K}^\star)^\top\boldsymbol{P}_i \\
        &+ \boldsymbol{C}_{cl}(\boldsymbol{K}^\star)^\top\boldsymbol{P}_i \boldsymbol{C}_{cl}(\boldsymbol{K}^\star)+\boldsymbol{Q}_i 
        + \sum_{j\in\mathcal{I}} \boldsymbol{K}_j^{\star\top}\boldsymbol{R}_{ij} \boldsymbol{K}_j^\star
    \end{aligned}
\end{equation}
    with
    \begin{equation}\label{eq:control_strategy}
    \begin{aligned}
        \boldsymbol{K}_i^\star &= \left(\boldsymbol{R}_{ii}+\boldsymbol{D}_i^\top\boldsymbol{P}_i\boldsymbol{D}_i\right)^{-1}\\&\bigg(\boldsymbol{B}_i^\top \boldsymbol{P}_i+\boldsymbol{D}_i^\top \boldsymbol{P}_i\boldsymbol{C}-\sum_{k\in\mathcal{I}\setminus\{i\}}\boldsymbol{D}_i^\top\boldsymbol{P}_i\boldsymbol{D}_k\boldsymbol{K}^\star_k\bigg).
    \end{aligned}
\end{equation}
Then, $\boldsymbol{K}^\star$ is a linear feedback Nash equilibrium as in Definition~\ref{def:nash_equilibrium} and $J_i(\boldsymbol{x}_0,\boldsymbol{K}^\star)=\boldsymbol{x}_0^\top\boldsymbol{P}_i\boldsymbol{x}_0,\, \forall i \in \mathcal{I}$. 

Conversely, if $\boldsymbol{K}^\star$ is a linear feedback Nash equilibrium, then the set of coupled stochastic algebraic Riccati equations \eqref{eq:coupled_riccati} with \eqref{eq:control_strategy}, $\forall i \in \mathcal{I}$, has a stabilizing solution.
\end{lemma}
\begin{pf}
(Sufficiency) Let there exist matrices $\boldsymbol{P}_i\in \mathbb{S}^n_{++},\,\forall i \in \mathcal{I}$, satisfying \eqref{eq:coupled_riccati} and \eqref{eq:control_strategy}. Consider player $i$ applying an arbitrary control of the form $\boldsymbol{u}_i(t) =\tilde{\boldsymbol{u}}_i(t) - \boldsymbol{K}_i^\star \boldsymbol{x}(t)$ and let other players apply $\boldsymbol{u}_j^\star(t)=-\boldsymbol{K}_j^\star \boldsymbol{x}(t),\, \forall j\in \mathcal{I}\setminus\{i\}$, such that \eqref{eq:system} is mean-square stable, i.e., \eqref{eq:mean_square_stability} holds.
Consider the quadratic value function candidate $V_i(\boldsymbol{x}(t)) = \boldsymbol{x}(t)^\top \boldsymbol{P}_i \boldsymbol{x}(t)$ and apply Itô's formula \cite[Theorem~1.3.1]{Dam:04} along the dynamics \eqref{eq:system} yielding
\begin{equation}\label{eq:suboptimal_value_function_increment}
\begin{aligned}
dV_i(\boldsymbol{x}(t))
&=
\bigg[
\boldsymbol{x}(t)^\top (\boldsymbol{P}_i \boldsymbol{A} + \boldsymbol{A}^\top \boldsymbol{P}_i + \boldsymbol{C}^\top \boldsymbol{P}_i \boldsymbol{C}) \boldsymbol{x}(t)\\
&+ 2 \sum_{j\in\mathcal{I}} \boldsymbol{x}(t)^\top (\boldsymbol{P}_i \boldsymbol{B}_j + \boldsymbol{C}^\top \boldsymbol{P}_i \boldsymbol{D}_j)\boldsymbol{u}_j(t)\\
&+ \sum_{j\in\mathcal{I}} \sum_{k\in\mathcal{I}}
\boldsymbol{u}_j(t)^\top \boldsymbol{D}_j^\top \boldsymbol{P}_i \boldsymbol{D}_k \boldsymbol{u}_k(t)
\bigg] dt \\
&+ 2 \boldsymbol{x}(t)^\top \boldsymbol{P}_i
\bigg(
\boldsymbol{C}\boldsymbol{x}(t) + \sum_{j\in\mathcal{I}} \boldsymbol{D}_j \boldsymbol{u}_j(t)
\bigg) d\omega(t).
\end{aligned}
\end{equation}
Adding the infinitesimal running cost of player $i$ from \eqref{eq:cost_function} and substituting $\boldsymbol{u}_j^\star(t)=-\boldsymbol{K}_j^\star \boldsymbol{x}(t),\, \forall j\in \mathcal{I}\setminus\{i\}$, and $\boldsymbol{u}_i(t) =\tilde{\boldsymbol{u}}_i(t) - \boldsymbol{K}_i^\star \boldsymbol{x}(t)$ into \eqref{eq:suboptimal_value_function_increment} gives
\begin{equation}\label{eq:expanded_suboptimal_value_function_increment}
    \begin{aligned}
        &dV_i(\boldsymbol{x}(t)) + \bigg( \boldsymbol{x}(t)^\top\boldsymbol{Q}_i\boldsymbol{x}(t) + \sum_{j\in\mathcal{I}} \boldsymbol{u}_j(t)^\top\boldsymbol{R}_{ij}\boldsymbol{u}_j(t) \bigg) dt\\
        &= \boldsymbol{x}(t)^\top\bigg[\boldsymbol{Q}_i + \boldsymbol{P}_i \boldsymbol{A}_{cl}(\boldsymbol{K}^\star) + \boldsymbol{A}_{cl}(\boldsymbol{K}^\star)^\top\boldsymbol{P}_i\\ 
        &+ \boldsymbol{C}_{cl}(\boldsymbol{K}^\star)^\top\boldsymbol{P}_i \boldsymbol{C}_{cl}(\boldsymbol{K}^\star)
        + \sum_{j\in\mathcal{I}} \boldsymbol{K}_j^{\star\top}\boldsymbol{R}_{ij} \boldsymbol{K}_j^\star\bigg]\boldsymbol{x}(t)dt \\
        &+ 2\tilde{\boldsymbol{u}}_i(t)^\top\bigg[\boldsymbol{B}_i^\top \boldsymbol{P}_i+\boldsymbol{D}_i^\top \boldsymbol{P}_i\boldsymbol{C} -\sum_{k\in\mathcal{I}\setminus\{i\}}\boldsymbol{D}_i^\top\boldsymbol{P}_i\boldsymbol{D}_k\boldsymbol{K}^\star_k\\
        &-\left(\boldsymbol{R}_{ii}+\boldsymbol{D}_i^\top\boldsymbol{P}_i\boldsymbol{D}_i\right)\boldsymbol{K}_i^\star\bigg]\boldsymbol{x}(t)dt \\
        &+\tilde{\boldsymbol{u}}_i(t)^\top 
        \left(\boldsymbol{R}_{ii} + \boldsymbol{D}_i^\top \boldsymbol{P}_i \boldsymbol{D}_i \right)
        \tilde{\boldsymbol{u}}_i(t)dt\\
        &+ 2\boldsymbol{x}(t)^\top \boldsymbol{P}_i 
        \left( \boldsymbol{C}_{cl}(\boldsymbol{K}^\star) \boldsymbol{x}(t) + \boldsymbol{D}_i \tilde{\boldsymbol{u}}_i(t) \right) d\omega(t).
    \end{aligned}
\end{equation}
Since \eqref{eq:coupled_riccati} and \eqref{eq:control_strategy} hold, the drift terms in \eqref{eq:expanded_suboptimal_value_function_increment} that are quadratic in $\boldsymbol{x}(t)$ and linear in $\tilde{\boldsymbol{u}}_i(t)$ vanish. Consequently,
\begin{equation}\label{eq:reduced_suboptimal_value_function_increment}
    \begin{aligned}
        &dV_i(\boldsymbol{x}(t)) + \bigg( \boldsymbol{x}(t)^\top\boldsymbol{Q}_i\boldsymbol{x}(t) + \sum_{j\in\mathcal{I}} \boldsymbol{u}_j(t)^\top\boldsymbol{R}_{ij}\boldsymbol{u}_j(t) \bigg) dt\\
        &= \tilde{\boldsymbol{u}}_i(t)^\top 
        \left(\boldsymbol{R}_{ii} + \boldsymbol{D}_i^\top \boldsymbol{P}_i \boldsymbol{D}_i \right)
        \tilde{\boldsymbol{u}}_i(t)dt\\
        &+ 2\,\boldsymbol{x}(t)^\top \boldsymbol{P}_i 
        \left( \boldsymbol{C}_{cl}(\boldsymbol{K}^\star) \boldsymbol{x}(t) + \boldsymbol{D}_i \tilde{\boldsymbol{u}}_i(t) \right) d\omega(t).
    \end{aligned}
\end{equation}
 For fixed $T>0$, we localize with the stopping time $\tau_l := \min(\inf\{t \geq 0: \|\boldsymbol{x}(t)\| \geq l\}, T)$, so that the stopped integral
 \begin{equation}
     \int_0^{\tau_l}2\,\boldsymbol{x}(t)^\top \boldsymbol{P}_i 
        \left( \boldsymbol{C}_{cl}(\boldsymbol{K}^\star) \boldsymbol{x}(t) + \boldsymbol{D}_i \tilde{\boldsymbol{u}}_i(t) \right) d\omega(t)
 \end{equation}
 is a true martingale with zero mean. Integrating \eqref{eq:reduced_suboptimal_value_function_increment} over $[0,\tau_l]$ and taking expectations thus eliminates the stochastic integral on this interval. Since $\tau_l\to T$ almost surely as $l\to\infty$, the stopped running cost integral
        \begin{equation}
            \int_0^{\tau_l}\boldsymbol{x}(t)^\top\boldsymbol{Q}_i\boldsymbol{x}(t) + \sum_{j\in\mathcal{I}} \boldsymbol{u}_j(t)^\top\boldsymbol{R}_{ij}\boldsymbol{u}_j(t)dt
        \end{equation}
        and the stopped deviation integral 
        \begin{equation}
            \int_0^{\tau_l}\tilde{\boldsymbol{u}}_i(t)^\top 
        (\boldsymbol{R}_{ii} + \boldsymbol{D}_i^\top \boldsymbol{P}_i \boldsymbol{D}_i)
        \tilde{\boldsymbol{u}}_i(t)dt
        \end{equation}
        increase almost surely to their counterparts on $[0,T]$, and the monotone convergence theorem \cite{Nis:15} applies. Moreover, continuity of the state trajectories implies $V_i(\boldsymbol{x}(\tau_l))\to V_i(\boldsymbol{x}(T))$ almost surely, and since $\mathbb{E}\bigl[\sup_{t\in[0,T]}\|\boldsymbol{x}(t)\|^2\bigr]<\infty$ for \eqref{eq:system}, the dominated convergence theorem \cite{Nis:15} implies $\mathbb{E}[V_i(\boldsymbol{x}(\tau_l))]\to\mathbb{E}[V_i(\boldsymbol{x}(T))]$. Hence,
\begin{equation}
\begin{aligned}
&\mathbb{E}\bigg[\int_{0}^T\boldsymbol{x}(t)^\top\boldsymbol{Q}_i\boldsymbol{x}(t) + \sum_{j\in\mathcal{I}} \boldsymbol{u}_j(t)^\top\boldsymbol{R}_{ij}\boldsymbol{u}_j(t)\,dt\bigg]\\ &=
V_i(\boldsymbol{x}_0) - \mathbb{E}\left[V_i(\boldsymbol{x}(T))\right]\\
&+ \mathbb{E} \bigg[\int_0^T 
\tilde{\boldsymbol{u}}_i(t)^\top 
\left(\boldsymbol{R}_{ii} + \boldsymbol{D}_i^\top \boldsymbol{P}_i \boldsymbol{D}_i \right)
\tilde{\boldsymbol{u}}_i(t) dt\bigg].
\end{aligned}
\end{equation}
Letting $T\rightarrow\infty$ and using that \eqref{eq:mean_square_stability} is assumed to hold yields
\begin{equation}\label{eq:suboptimal_cost_function}
\begin{aligned}
&J_i(\boldsymbol{x}_0, \boldsymbol{u}_i(t), \boldsymbol{u}_{\neg i}^\star(t)) =
V_i(\boldsymbol{x}_0)\\
&+ \mathbb{E} \left[\int_0^\infty 
\tilde{\boldsymbol{u}}_i(t)^\top 
\left(\boldsymbol{R}_{ii} + \boldsymbol{D}_i^\top \boldsymbol{P}_i \boldsymbol{D}_i \right)
\tilde{\boldsymbol{u}}_i(t) dt\right].
\end{aligned}
\end{equation}
Under Assumption~\ref{ass:admissible_set}, and since $\boldsymbol{P}_i\in \mathbb{S}^n_{++}$, we get $\boldsymbol{R}_{ii}+\boldsymbol{D}_i^\top\boldsymbol{P}_i\boldsymbol{D}_i\in \mathbb{S}^{m_i}_{++}$. Thus \eqref{eq:suboptimal_cost_function} is minimized if and only if $\tilde{\boldsymbol{u}}_i(t) \equiv 0$, i.e., $\boldsymbol{u}_i(t) = -\boldsymbol{K}_i^\star \boldsymbol{x}(t)$. Since this holds for all players, $\boldsymbol{K}^\star$ constitutes a linear feedback Nash equilibrium with $J_i(\boldsymbol{x}_0,\boldsymbol{K}^\star)=\boldsymbol{x}_0^\top\boldsymbol{P}_i\boldsymbol{x}_0,\, \forall i \in \mathcal{I}$.

(Necessity) Let $\boldsymbol{K}^\star$ be a linear feedback Nash equilibrium according to Definition~\ref{def:nash_equilibrium}. Then, for each $i \in \mathcal{I}$, the associated value function is quadratic in $\boldsymbol{x}(t)$, and there exists a unique stabilizing matrix $\boldsymbol{P}_i \in \mathbb{S}^n_{++}$ such that $V_i(\boldsymbol{x}(t))=\boldsymbol{x}(t)^\top\boldsymbol{P}_i\boldsymbol{x}(t)$ \cite[Lemma~1]{Han:26}. Consider $i \in \mathcal{I}$ with $\boldsymbol{u}_i(t)$ and suppose that all other players apply their fixed equilibrium strategies $\boldsymbol{u}_j^\star(t)=-\boldsymbol{K}_j^\star\boldsymbol{x}(t),\, \forall j \in \mathcal{I}\setminus\{i\}$, such that \eqref{eq:system} is mean-square stable, i.e., \eqref{eq:mean_square_stability} holds. We then consider the resulting single-player stochastic optimal control problem for player $i$. By the dynamic programming principle \cite{Bel:66} and Itô's formula \cite[Theorem~1.3.1]{Dam:04}, the corresponding stochastic \gls{hjb} equation for player $i$ is given by
\begin{equation}\label{eq:best_response_hjb}
\begin{aligned}
0 &= \min_{\boldsymbol{u}_i(t)} \bigg\{
\boldsymbol{x}(t)^\top \boldsymbol{Q}_i \boldsymbol{x}(t)
+ \sum_{j\in\mathcal{I}\setminus\{i\}} \boldsymbol{x}(t)^\top \boldsymbol{K}_j^{\star\top} \boldsymbol{R}_{ij} \boldsymbol{K}_j^\star \boldsymbol{x}(t)\\
&+ \boldsymbol{u}_i(t)^\top \boldsymbol{R}_{ii} \boldsymbol{u}_i(t) \\
&+ 2 \boldsymbol{x}(t)^\top \boldsymbol{P}_i 
\bigg(
\boldsymbol{A} \boldsymbol{x}(t) -\sum_{j\in\mathcal{I}\setminus\{i\}} \boldsymbol{B}_j \boldsymbol{K}_j^\star \boldsymbol{x}(t)
+ \boldsymbol{B}_i \boldsymbol{u}_i(t)
\bigg) \\
&+ \operatorname{tr}\bigg[\bigg(
\boldsymbol{C} \boldsymbol{x}(t) - \sum_{j\in\mathcal{I}\setminus\{i\}} \boldsymbol{D}_j \boldsymbol{K}_j^\star \boldsymbol{x}(t)
+ \boldsymbol{D}_i \boldsymbol{u}_i(t)
\bigg)^\top \boldsymbol{P}_i \\
&\bigg(
\boldsymbol{C} \boldsymbol{x}(t) - \sum_{k\in\mathcal{I}\setminus\{i\}} \boldsymbol{D}_k \boldsymbol{K}_k^\star \boldsymbol{x}(t)
+ \boldsymbol{D}_i \boldsymbol{u}_i(t)
\bigg)\bigg]
\bigg\}.
\end{aligned}
\end{equation}
Differentiating \eqref{eq:best_response_hjb} with respect to $\boldsymbol{u}_i(t)$ and setting the derivative to zero yields
\begin{equation}\label{eq:first_order_condition_br}
\begin{aligned}
\boldsymbol{0}
&= 2\left(\boldsymbol{R}_{ii} + \boldsymbol{D}_i^\top \boldsymbol{P}_i \boldsymbol{D}_i\right)\boldsymbol{u}_i(t) \\
& + 2\bigg(
\boldsymbol{B}_i^\top \boldsymbol{P}_i 
+ \boldsymbol{D}_i^\top \boldsymbol{P}_i \boldsymbol{C}
- \sum_{k\in\mathcal{I}\setminus\{i\}} \boldsymbol{D}_i^\top \boldsymbol{P}_i \boldsymbol{D}_k \boldsymbol{K}_k^\star
\bigg)\boldsymbol{x}(t).
\end{aligned}
\end{equation}
The Hessian of \eqref{eq:best_response_hjb} with respect to $\boldsymbol{u}_i(t)$ is given by $2\left(\boldsymbol{R}_{ii} + \boldsymbol{D}_i^\top \boldsymbol{P}_i \boldsymbol{D}_i\right)$, which is positive definite under Assumption~\ref{ass:admissible_set} and since $\boldsymbol{P}_i\in \mathbb{S}^n_{++}$, ensuring strict convexity and uniqueness of the minimizer, which by \eqref{eq:first_order_condition_br} is of linear feedback form. Since $\boldsymbol{K}^\star$ is a linear feedback Nash equilibrium, $\boldsymbol{u}_i^\star(t)=-\boldsymbol{K}_i^\star\boldsymbol{x}(t)$ is an optimal control for this same single-player problem. Hence, $\boldsymbol{u}_i^\star(t)$ coincides with the solution of \eqref{eq:first_order_condition_br} by the uniqueness of the minimizer, and \eqref{eq:control_strategy} follows. Substituting this control into \eqref{eq:best_response_hjb} and collecting all terms quadratic in $\boldsymbol{x}(t)$ yields \eqref{eq:coupled_riccati}. Since this holds for all players given $\boldsymbol{K}^\star$, there exist matrices $\boldsymbol{P}_i\in\mathbb{S}_{++}^n$, $i \in \mathcal{I}$, such that \eqref{eq:coupled_riccati} and \eqref{eq:control_strategy} hold.\qed
\end{pf}
\begin{remark}\label{rem:compare_to_det_condition}
Comparing Lemma~\ref{lem:coupled_Riccati_necessary_sufficient} with the corresponding result for deterministic infinite-horizon differential games in \cite[Theorem~4]{Eng:00}, as well as with \gls{lqg} differential games and stochastic \gls{lq} differential games with purely state-dependent noise, the equilibrium feedback strategies \eqref{eq:control_strategy} are fundamentally altered by the control-dependent noise. In the aforementioned settings, \eqref{eq:control_strategy} reduces to $\boldsymbol{K}_i^\star=\boldsymbol{R}_{ii}^{-1}\boldsymbol{B}_i^\top\boldsymbol{P}_i$, such that the algebraic Riccati equations are coupled only through $\sum_{j\in\mathcal{I}}\boldsymbol{K}_j^{\star\top}\boldsymbol{R}_{ij}\boldsymbol{K}_j^\star$ in \eqref{eq:coupled_riccati}, yielding a system of polynomial matrix equations. By contrast, the control-dependent noise introduces the term $\sum_{k\in\mathcal{I}\setminus\{i\}}\boldsymbol{D}_i^\top\boldsymbol{P}_i\boldsymbol{D}_k\boldsymbol{K}_k^\star$ in \eqref{eq:control_strategy}, such that each player's equilibrium strategy depends explicitly on the equilibrium strategies $\boldsymbol{K}_k^\star$ of the other players. Consequently, the coupled stochastic algebraic Riccati equations \eqref{eq:coupled_riccati} and \eqref{eq:control_strategy} form a system of rational matrix equations.
\end{remark}
In view of the structural differences discussed in Remark~\ref{rem:compare_to_det_condition}, Lemma~\ref{lem:coupled_Riccati_necessary_sufficient} provides the theoretical prerequisite for the proposed \gls{isdg} method.
\subsection{Kernel Representation of Equilibrium Conditions}
Since Lemma~\ref{lem:coupled_Riccati_necessary_sufficient} establishes necessary and sufficient conditions for linear feedback Nash equilibria, the cost function matrices $\boldsymbol{Q}_i$ and $\boldsymbol{R}_{ij},\,\forall i,j \in \mathcal{I}$, must satisfy \eqref{eq:coupled_riccati} with \eqref{eq:control_strategy} for the recovered strategies $\hat{\boldsymbol{K}}_i = \boldsymbol{K}_i^\star,\,\forall i \in \mathcal{I}$, as established in Lemma~\ref{lem:identification_strategies}. We exploit this property to develop an inverse method for estimating the cost function parameters by deriving a kernel representation of \eqref{eq:coupled_riccati} and \eqref{eq:control_strategy}, as described in the following.
\begin{lemma}\label{lem:kernel_Riccati}
    Let Assumption~\ref{ass:mss} hold. Then, $\forall i \in \mathcal{I}$, the matrices $\boldsymbol{Q}_i$ and $\boldsymbol{R}_{ij},\, \forall j \in \mathcal{I}$, contained in $\boldsymbol{\theta}_i$ as defined in \eqref{eq:theta_i}, satisfy \eqref{eq:coupled_riccati} and \eqref{eq:control_strategy} if and only if $\boldsymbol{\theta}_i$ satisfies
\begin{equation}\label{eq:kernel_riccati}
    \boldsymbol{M}_i\boldsymbol{\theta}_i = \boldsymbol{0} \quad \Leftrightarrow \quad \boldsymbol{\theta}_i \in \ker(\boldsymbol{M}_i).
\end{equation}
    The matrix $\boldsymbol{M}_i\in \mathbb{R}^{nm_i \times p_i}$ in \eqref{eq:kernel_riccati} is defined as
    \begin{equation}\label{eq:matrix_M_i}
    \begin{aligned}
         \boldsymbol{M}_i &:= \Big[
        \boldsymbol{S}_i
       \quad
       \boldsymbol{S}_i\boldsymbol{K}_1^{\otimes}
       \;\;
       \cdots
       \;\;
       \boldsymbol{S}_i\boldsymbol{K}_{i-1}^{\otimes}\quad
       \Big(\left(\boldsymbol{K}_i^{\top} \otimes \boldsymbol{I}_{m_i}\right)\\
       &+\boldsymbol{S}_i \boldsymbol{K}_i^{\otimes}\Big)\quad
       \boldsymbol{S}_i\boldsymbol{K}_{i+1}^{\otimes}
       \;\;
       \cdots
       \;\;
       \boldsymbol{S}_i \boldsymbol{K}_N^{\otimes}
        \Big],
    \end{aligned}
    \end{equation}
    with 
    \begin{equation}\label{eq:definition_of_K_i_otimes}
        \boldsymbol{K}_i^{\otimes}:=\boldsymbol{K}_i^{\top} \otimes  \boldsymbol{K}_i^{\top}\in\mathbb{R}^{n^2 \times m_i^2}
    \end{equation}
    and
    \begin{equation}\label{eq:definition_of_S_i}
        \begin{aligned}
            \boldsymbol{S}_i&:= \bigg(\left(\boldsymbol{I}_n \otimes \boldsymbol{B}_i^\top\right)+\left(\boldsymbol{C}^\top \otimes \boldsymbol{D}_i^\top\right)\\
            &-\sum_{k\in\mathcal{I}}\left(\boldsymbol{K}^{\top}_k\boldsymbol{D}_k^\top \otimes \boldsymbol{D}_i^\top\right)\bigg)\\
            &\Big(\boldsymbol{A}_{cl}(\boldsymbol{K})^\top \otimes \boldsymbol{I}_n
            +\boldsymbol{I}_n  \otimes \boldsymbol{A}_{cl}(\boldsymbol{K})^\top\\
            &+\boldsymbol{C}_{cl}(\boldsymbol{K})^\top \otimes \boldsymbol{C}_{cl}(\boldsymbol{K})^\top\Big)^{-1}\in \mathbb{R}^{nm_i \times n^2},
        \end{aligned}
    \end{equation}
    where $\otimes$ denotes the Kronecker product. 
\end{lemma}
\begin{pf}
    Applying the vectorization identity
    \begin{equation}\label{eq:vectorization_identity}
        \operatorname{vec}(\boldsymbol{X}\boldsymbol{Y}\boldsymbol{Z})
        = (\boldsymbol{Z}^\top \otimes \boldsymbol{X}) \operatorname{vec}(\boldsymbol{Y}),        
    \end{equation}
    we rewrite \eqref{eq:coupled_riccati} as
    \begin{equation}\label{eq:vectorized_riccati}
        \begin{aligned}
            \operatorname{vec}(\boldsymbol{P}_i)&=-\Big(\boldsymbol{A}_{cl}(\boldsymbol{K})^\top \otimes \boldsymbol{I}_n
            +\boldsymbol{I}_n  \otimes \boldsymbol{A}_{cl}(\boldsymbol{K})^\top\\
            &+\boldsymbol{C}_{cl}(\boldsymbol{K})^\top \otimes \boldsymbol{C}_{cl}(\boldsymbol{K})^\top\Big)^{-1}\\&\bigg(\operatorname{vec}\left(\boldsymbol{Q}_i\right) 
            +\sum_{j\in\mathcal{I}} \left(\boldsymbol{K}_j^{\top} \otimes  \boldsymbol{K}_j^{\top}\right)\operatorname{vec}\left(\boldsymbol{R}_{ij}\right)\bigg).
\end{aligned}
    \end{equation}
    By Assumption~\ref{ass:mss}, \eqref{eq:closed_loop_system} is mean-square stable. Using \cite[Theorem~1.5.3]{Dam:04} and the Kronecker representation, this is equivalent to
    \begin{equation}
    \begin{aligned}
         &\sigma\Big(\boldsymbol{A}_{cl}(\boldsymbol{K})^\top \otimes \boldsymbol{I}_n
    +\boldsymbol{I}_n  \otimes \boldsymbol{A}_{cl}(\boldsymbol{K})^\top\\
    &+\boldsymbol{C}_{cl}(\boldsymbol{K})^\top \otimes \boldsymbol{C}_{cl}(\boldsymbol{K})^\top\Big)
        \subset \mathbb{C}_-,
    \end{aligned}
    \end{equation}
    where $\sigma(\cdot)$ denotes the spectrum and $\mathbb{C}_-$ the left half of the complex plane.  
    In particular, zero is not contained in the spectrum, and therefore the inverse appearing in \eqref{eq:vectorized_riccati} exists. Next, applying \eqref{eq:vectorization_identity} to \eqref{eq:control_strategy} yields
    \begin{equation}\label{eq:vectorized_strategy}
    \begin{aligned}
        \boldsymbol{0} &=\left(\boldsymbol{K}_i^{\top} \otimes \boldsymbol{I}_{m_i}\right)\operatorname{vec}\left(\boldsymbol{R}_{ii}\right)- \left(\boldsymbol{I}_n \otimes \boldsymbol{B}_i^\top\right) \operatorname{vec}\left(\boldsymbol{P}_i\right)\\
        &-\left(\boldsymbol{C}^\top \otimes \boldsymbol{D}_i^\top\right)\operatorname{vec}\left(\boldsymbol{P}_i\right)\\
        &+\sum_{k\in\mathcal{I}}\left(\boldsymbol{K}^{\top}_k\boldsymbol{D}_k^\top \otimes \boldsymbol{D}_i^\top\right)\operatorname{vec}\left(\boldsymbol{P}_i\right).
    \end{aligned}
    \end{equation}
    Inserting \eqref{eq:vectorized_riccati} into \eqref{eq:vectorized_strategy}, rearranging the resulting expression, and applying the definitions in \eqref{eq:matrix_M_i}, \eqref{eq:definition_of_K_i_otimes}, \eqref{eq:definition_of_S_i}, and \eqref{eq:theta_i} yields \eqref{eq:kernel_riccati}. \qed
\end{pf}
\begin{remark}
    Lemma~\ref{lem:kernel_Riccati} strictly generalizes the kernel representation underlying inverse methods for deterministic infinite-horizon \gls{lq} differential games, e.g.,~\cite{Ing:19}, to the setting with state- and control-dependent noise. In particular, the deterministic case arises as the special case $\boldsymbol{C}=\boldsymbol{0}$, $\boldsymbol{D}_i=\boldsymbol{0},\,\forall i\in\mathcal{I}$, in \eqref{eq:system}, for which the kernel representation \eqref{eq:kernel_riccati} continues to hold.
\end{remark}
\subsection{Solution Sets for Inverse Stochastic Differential Games}
In general, an infinite-horizon \gls{lq} differential game with state- and control-dependent noise may admit multiple solutions to \eqref{eq:coupled_riccati} with \eqref{eq:control_strategy}, corresponding to different Nash equilibria, similarly to phenomena observed in the deterministic case \cite{Thoe:26}. In the inverse setting, however, Assumption~\ref{ass:nash_trajectories} fixes the specific Nash equilibrium $\boldsymbol{K}^\star$ that is consistent with the observed demonstrations, yet an intrinsic ambiguity of the inverse problem remains. Specifically, multiple cost function parameter vectors $\boldsymbol{\theta}_i$ satisfy \eqref{eq:coupled_riccati} and \eqref{eq:control_strategy} for the same fixed $\boldsymbol{K}^\star$. To characterize all cost function parameterizations consistent with the observed equilibrium trajectories and thereby solve Problem~\ref{prob:isdg}, we now present our main result.
\begin{theorem}\label{theo:admissible_set}
    Let a stochastic differential game be given by \eqref{eq:system} and \eqref{eq:cost_function}. Let Assumptions~\ref{ass:mss},~\ref{ass:admissible_set},~\ref{ass:nash_trajectories}, and~\ref{ass:persistence_excitation} hold. Then the solution set as in Definition~\ref{def:solution_set}, corresponding to the given demonstrations, is determined as
    \begin{equation}\label{eq:admissible_set}
        \boldsymbol{\Theta}=\prod_{i\in\mathcal{I}} \left(\ker(\boldsymbol{M}_i) \cap \mathcal{A}_i\right).
    \end{equation}
\end{theorem}
\begin{pf}
    By Assumptions~\ref{ass:mss},~\ref{ass:nash_trajectories}, and~\ref{ass:persistence_excitation}, Lemma~\ref{lem:identification_strategies} yields $\hat{\boldsymbol{K}}_i = \boldsymbol{K}_i^\star$, $\hat{\boldsymbol{x}}(t)=\boldsymbol{x}^\star(t)$, and $\hat{\boldsymbol{u}}_i(t)=\boldsymbol{u}_i^\star(t), \, \forall i \in \mathcal{I}$, almost surely. Hence, we set $\boldsymbol{K}^\star=(\boldsymbol{K}_1^\star,\dots,\boldsymbol{K}_N^\star)$ for what follows.
    
    (Necessity)
    Let $\hat{\boldsymbol{\theta}}=(\hat{\boldsymbol{\theta}}_1,\dots,\hat{\boldsymbol{\theta}}_N)\in\boldsymbol{\Theta}$ as in Definition~\ref{def:solution_set}. By Definition~\ref{def:solution_set} and under Assumption~\ref{ass:admissible_set}, $\hat{\boldsymbol{\theta}}$ is admissible and induces $\boldsymbol{K}^\star$.
     By Lemma~\ref{lem:coupled_Riccati_necessary_sufficient} and under Assumptions~\ref{ass:mss} and~\ref{ass:admissible_set}, there exist matrices $\boldsymbol{P}_i \in \mathbb{S}^n_{++},\,\forall i \in \mathcal{I}$, such that \eqref{eq:coupled_riccati} and \eqref{eq:control_strategy} hold with $\boldsymbol{K}^\star$. Under Assumption~\ref{ass:mss}, Lemma~\ref{lem:kernel_Riccati} implies equivalently $\hat{\boldsymbol{\theta}}_i \in \ker(\boldsymbol{M}_i), \, \forall i \in \mathcal{I}$, and further under Assumption~\ref{ass:admissible_set}, $\hat{\boldsymbol{\theta}}_i \in \ker(\boldsymbol{M}_i) \cap \mathcal{A}_i, \, \forall i \in \mathcal{I}$. Therefore,
    \begin{equation}\label{eq:set_direction_1}
        \hat{\boldsymbol{\theta}} \in \Theta \Rightarrow \hat{\boldsymbol{\theta}} \in \prod_{i\in\mathcal{I}} \left(\ker(\boldsymbol{M}_i) \cap \mathcal{A}_i\right).
    \end{equation}
    
    (Sufficiency)
    Let 
    \begin{equation}\label{eq:theta_hat_in_prod}
        \hat{\boldsymbol{\theta}}=(\hat{\boldsymbol{\theta}}_1,\dots,\hat{\boldsymbol{\theta}}_N) \in \prod_{i\in\mathcal{I}} (\ker(\boldsymbol{M}_i)\cap\mathcal{A}_i).
    \end{equation}
    By Lemma~\ref{lem:kernel_Riccati} and under Assumption~\ref{ass:mss}, $\hat{\boldsymbol{\theta}}_i \in \ker(\boldsymbol{M}_i)$,
    $\forall i \in \mathcal{I}$, is equivalent to the satisfaction of \eqref{eq:coupled_riccati} and \eqref{eq:control_strategy} at the fixed strategies $\boldsymbol{K}^\star$. Moreover, by admissibility via Assumption~\ref{ass:admissible_set} and by \cite[Lemma~1]{Han:26}, there exist matrices $\boldsymbol{P}_i \in \mathbb{S}_{++}^n,\,\forall i\in\mathcal{I}$, and it follows by the sufficiency of Lemma~\ref{lem:coupled_Riccati_necessary_sufficient} that $\boldsymbol{K}^\star$ is a linear feedback Nash equilibrium induced by $\hat{\boldsymbol{\theta}}$. Hence, 
    by Definition~\ref{def:solution_set}, $\hat{\boldsymbol{\theta}}\in \boldsymbol{\Theta}$. That is,
    \begin{equation}\label{eq:set_direction_2}
        \hat{\boldsymbol{\theta}} \in \prod_{i\in\mathcal{I}} (\ker(\boldsymbol{M}_i)\cap\mathcal{A}_i) \Rightarrow \hat{\boldsymbol{\theta}} \in \Theta.
    \end{equation}
    
    Since \eqref{eq:set_direction_1} and \eqref{eq:set_direction_2} both hold, we conclude \eqref{eq:admissible_set}.\qed
\end{pf}
\begin{remark}\label{rem:non_uniqueness}
    The solution set $\boldsymbol{\Theta}$, determined by \eqref{eq:admissible_set}, ensures that $\boldsymbol{K}^\star$ is a linear feedback Nash equilibrium for all $\hat{\boldsymbol{\theta}} \in \boldsymbol{\Theta}$, but does not guarantee its uniqueness. For a fixed joint parameter vector $\hat{\boldsymbol{\theta}}$, the corresponding game may admit multiple Nash equilibria, of which $\boldsymbol{K}^\star$ is only one.
\end{remark}
\begin{corollary}\label{cor:set_not_empty}
    Under Assumptions~\ref{ass:mss},~\ref{ass:admissible_set},~\ref{ass:nash_trajectories}, and~\ref{ass:persistence_excitation}, the solution set satisfies $\boldsymbol{\Theta} \neq \varnothing$.
\end{corollary}
\begin{pf}
    By Assumptions~\ref{ass:mss},~\ref{ass:nash_trajectories}, and~\ref{ass:persistence_excitation}, the ground truth trajectories correspond to a linear feedback Nash equilibrium $\boldsymbol{K}^\star$ induced by the ground truth parameters $\boldsymbol{\theta}^\star$. By Lemmas~\ref{lem:coupled_Riccati_necessary_sufficient} and~\ref{lem:kernel_Riccati}, this implies $\boldsymbol{\theta}_i^\star \in \ker(\boldsymbol{M}_i)$, $\forall i \in \mathcal{I}$, and by Assumption~\ref{ass:admissible_set}, $\boldsymbol{\theta}_i^\star \in \mathcal{A}_i$. Hence, $\boldsymbol{\theta}^\star \in \boldsymbol{\Theta}$, and thus $\boldsymbol{\Theta} \neq \varnothing$.\qed
\end{pf}
Theorem~\ref{theo:admissible_set} solves Problem~\ref{prob:isdg} by determining the solution set $\boldsymbol{\Theta}$ in \eqref{eq:admissible_set} according to Definition~\ref{def:solution_set}, i.e., the set of all joint cost function parameters that induce $\boldsymbol{K}^\star$ as a linear feedback Nash equilibrium, which by Corollary~\ref{cor:set_not_empty} is guaranteed to be non-empty. Since $\boldsymbol{\Theta}$ is given explicitly as an intersection of a kernel and an admissible set for each player, any $\hat{\boldsymbol{\theta}} \in \boldsymbol{\Theta}$ can be obtained directly by sampling from $\boldsymbol{\Theta}$, without further optimization or search. This establishes an inverse method that identifies all cost function parameters of an infinite-horizon \gls{lq} differential game with state- and control-dependent noise consistent with observed equilibrium trajectories.
\section{Simulation Example}
To illustrate the proposed \gls{isdg} method, we consider a two-player scalar example system, allowing us to visualize the theoretical results of Theorem~\ref{theo:admissible_set} directly in the three-dimensional space of each player's cost function parameters. The dynamics of the example system are given by
\begin{equation}
\begin{aligned}
    dx(t)&=(x(t)+u_1(t)+u_2(t))dt\\
    &+(0.2x(t)+0.3u_1(t)+0.2u_2(t))d\omega(t),
\end{aligned}
\end{equation}
with initial state $x_0 = -1$. We first apply the policy iteration method from \cite{Han:26} with the cost function matrices $Q^\star_1=2$, $R^\star_{11}=1$, $R^\star_{12}=0.5$, $Q^\star_2=1.5$, $R^\star_{21}=0.5$, and $R^\star_{22}=1$ to calculate the linear feedback Nash equilibrium $\boldsymbol{K}^\star=(1.3515,1.2080)$. Next, we simulate the closed-loop system under $\boldsymbol{K}^\star$ to produce $D=20$ stochastic ground truth trajectories. To verify that $\boldsymbol{K}^\star$ is indeed a linear feedback Nash equilibrium, the deviations from the equilibrium conditions \eqref{eq:coupled_riccati} with \eqref{eq:control_strategy} for both players are combined into a residual $r$, yielding $r^\star=1.06 \times 10^{-12}$ for $\boldsymbol{K}^\star$. Moreover, we confirm algebraically that the resulting closed-loop system is mean-square stable, and the ground truth trajectories ensure persistence of excitation, such that Assumptions~\ref{ass:mss},~\ref{ass:admissible_set},~\ref{ass:nash_trajectories}, and~\ref{ass:persistence_excitation} hold.

We use the ground truth trajectories as the starting point for our \gls{isdg} method and obtain $\hat{\boldsymbol{K}}=(1.3515,1.2080)$ via \eqref{eq:ls_estimator}, which exactly matches $\boldsymbol{K}^\star$. Next, we use $\hat{\boldsymbol{K}}$ to construct $\boldsymbol{M}_i\in \mathbb{R}^{1\times3}$ as in \eqref{eq:matrix_M_i} for both players. The kernels of the resulting matrices are two-dimensional subspaces of $\mathbb{R}^3$ and can be parametrized as
\begin{equation}\label{eq:kernel_parameterization}
    \ker(\boldsymbol{M}_i) = \left\{ \alpha_i \boldsymbol{v}_{i,1} + \beta_i \boldsymbol{v}_{i,2} \;\middle|\; \alpha_i, \beta_i \in \mathbb{R} \right\},
\end{equation}
with
\begin{equation}
    \begin{aligned}
        \boldsymbol{v}_{1,1}&=[0.8392,0.4614,0.2880]^\top\\
        \boldsymbol{v}_{1,2}&=[-0.4486,0.2880,0.8461]^\top\\
        \boldsymbol{v}_{2,1}&=[-0.5730,0.7501,0.3302]^\top\\
        \boldsymbol{v}_{2,2}&=[0.7571,0.3302,0.5637]^\top.
    \end{aligned}
\end{equation}
To illustrate the player-wise solution sets defined by $\ker(\boldsymbol{M}_i) \cap \mathcal{A}_i$, we visualize in Figure~\ref{fig:kernels} the kernels $\ker(\boldsymbol{M}_i)$ for each player $i \in \{1,2\}$ in the respective parameter spaces $(Q_1,R_{11},R_{12})$ and $(Q_2,R_{21},R_{22})$. Following \eqref{eq:admissible_set} and \eqref{eq:kernel_parameterization}, each set is constructed numerically by sampling $\alpha_i,\beta_i \in [-5,5]$, which yields a finite representation of $\ker(\boldsymbol{M}_i)$. From these samples, the player-wise solution sets $\ker(\boldsymbol{M}_i)\cap\mathcal{A}_i$ are obtained by retaining only those parameter vectors that satisfy Assumption~\ref{ass:admissible_set}. The ground truth parameters lie within these player-wise solution sets, given by $\alpha_1^\star = 2.284$, $\beta_1^\star =-0.186$, $\alpha_2^\star = -0.154$, and $\beta_2^\star = 1.864$.
\begin{figure}
    \centering
       \begin{subfigure}{0.95\columnwidth}
        \centering
            \begin{tikzpicture}

\begin{axis}[
width=6.5cm,
xmin=-6,
xmax=6,
xlabel={$Q_{1}$},
ymin=-5,
ymax=5,
ylabel={$R_{11}$},
zmin=-5,
zmax=5,
zlabel={$R_{12}$},
view={380}{29},
xmajorgrids,
ymajorgrids,
zmajorgrids,
]

\addplot3[area legend, draw=none, fill=green!20!blue, fill opacity=0.5, forget plot]
table[row sep=crcr] {%
x	y	z\\
1.952765	3.746614	5.670061\\
-6.439019	-0.8671032	2.790551\\
-1.952765	-3.746614	-5.670061\\
6.439019	0.8671032	-2.790551\\
}--cycle;

\addplot3[area legend, draw=none, fill=black!20!red, fill opacity=0.8, forget plot]
table[row sep=crcr] {%
x	y	z\\
0	0	0\\
0.003158985	2.67474	5.001084\\
1.952765	3.746614	5.670061\\
4.943601	1.82694	0.0296533\\
}--cycle;

\addplot3[
    only marks,
    mark=*,
    mark size=2.0pt,
    mark options={fill=yellow, draw=yellow}
] coordinates {
    (2,1,0.5)
};

\end{axis}
\end{tikzpicture}%
            \caption{Player 1}
        \label{fig:p1}
    \end{subfigure}
    
    \vspace{1em}
    \begin{subfigure}{0.95\columnwidth}
        \centering
            \begin{tikzpicture}

\begin{axis}[
width=6.5cm,
xmin=-6,
xmax=6,
ymin=-5,
ymax=5,
zmin=-5,
zmax=5,
xlabel={$Q_{2}$},
ylabel={$R_{21}$},
zlabel={$R_{22}$},
view={380}{29},
xmajorgrids,
ymajorgrids,
zmajorgrids,
]

\addplot3[area legend, draw=none, fill=green!20!blue, fill opacity=0.5, forget plot]
table[row sep=crcr] {%
x	y	z\\
-0.92043   -5.40147   -4.469552\\
-6.650671   2.099108  -1.167187\\
0.92043     5.40147    4.469552\\
6.650671   -2.099108   1.167187\\
}--cycle;

\addplot3[area legend, draw=none, fill=black!20!red, fill opacity=0.8, forget plot]
table[row sep=crcr] {%
x	y	z\\
0           0           0\\
0.002720861 5.001187    3.786311\\
0.92043     5.40147     4.469552\\
5.029997    0.02228832  2.101189\\
}--cycle;

\addplot3[
    only marks,
    mark=*,
    mark size=2.0pt,
    mark options={fill=yellow, draw=yellow}
] coordinates {
    (1.5,0.5,1.0)
};

\end{axis}
\end{tikzpicture}
            \caption{Player 2}
        \label{fig:p2}
    \end{subfigure}
    \caption{The kernel of each matrix $\boldsymbol{M}_i$ forms a two-dimensional plane (blue and red). The player-wise solution sets $\ker(\boldsymbol{M}_i)\cap\mathcal{A}_i$ are shown in red. The ground truth parameters (yellow) lie within $\ker(\boldsymbol{M}_i)\cap\mathcal{A}_i$.}
    \label{fig:kernels}
\end{figure}
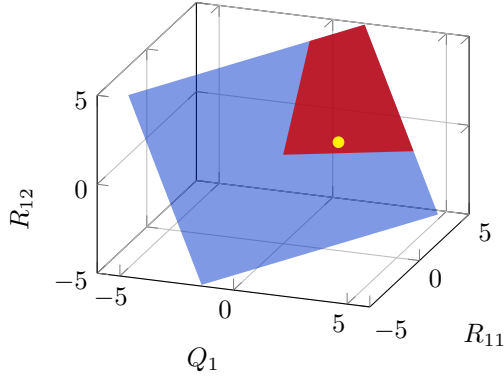
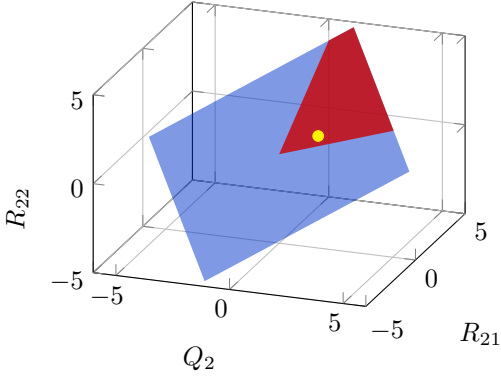
To validate Theorem~\ref{theo:admissible_set}, we construct the Cartesian product defining $\boldsymbol{\Theta}$ as in \eqref{eq:admissible_set}. For each player, $1000$ samples are drawn from $\ker(\boldsymbol{M}_i)\cap\mathcal{A}_i$, none of which coincides with the ground truth parameters. Each sampled parameter vector of player $i$ is then combined with all $1000$ samples of the other player, resulting in $10^6$ parameter tuples. For every such combination, admissibility is preserved by construction, and the residual $r$ is evaluated at the identified feedback strategy $\hat{\boldsymbol{K}}$. The maximal residual over all sampled combinations is $r_{\max} = 2.47 \times 10^{-12}$, confirming numerical consistency with the equilibrium conditions. Hence, all sampled admissible parameter combinations reproduce the equilibrium feedback strategies $\boldsymbol{K}^\star$. 
\begin{figure}
    \centering
    \input{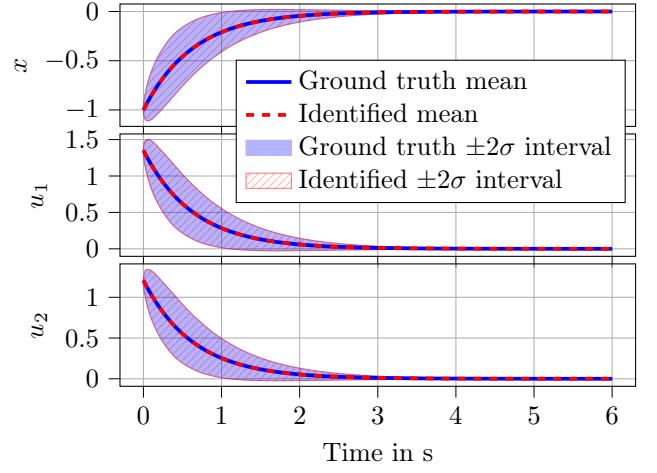}
     \caption{Ground truth and identified state and control trajectories with mean and $\pm2\sigma$ confidence intervals for $D=20$ demonstrations.}
    \label{fig:trajectories}
\end{figure}
This ensures trajectory matching in the sense of Definition~\ref{def:trajectory_matching}, as illustrated in Figure~\ref{fig:trajectories}, where the means and corresponding $\pm2\sigma$ confidence intervals of ground truth and identified state and control trajectories coincide. This confirms numerically that the proposed characterization of $\boldsymbol{\Theta}$ resolves Problem~\ref{prob:isdg}.
\section{Conclusion}
We presented a method to solve the inverse problem for infinite-horizon \gls{lq} differential games with state- and control-dependent noise. The proposed approach yields an explicit characterization of all cost function parameter combinations across players that admit the identified linear feedback Nash equilibrium consistent with observed equilibrium trajectories. Numerical results confirm the theoretical findings and highlight the inherent ambiguity of the inverse problem.

\section*{Declaration of Generative AI and AI-assisted technologies in the writing process}
During the preparation of this work the authors used ChatGPT in order to perform language editing. After using this tool, the authors reviewed and edited the content as needed and take full responsibility for the content of the publication.

\bibliographystyle{plain}        
\bibliography{autosam}           
\end{document}